%% file: main.tex
\documentclass{article}

\usepackage[letterpaper,top=2cm,bottom=2cm,left=3cm,right=3cm,marginparwidth=1.75cm]{geometry}
\usepackage{graphicx} 
\usepackage{amsmath}
\usepackage[colorlinks=true, allcolors=blue]{hyperref}
\usepackage{dirtytalk}
\usepackage{amsthm}
\usepackage{mathtools}
\usepackage{tabularx}
\usepackage{booktabs}
\usepackage{array}
\usepackage{amssymb}
\usepackage[T1]{fontenc}
\usepackage{thmtools}
\usepackage{thm-restate}

\title{A Note on EFX Inapproximability for Chores}
\author{
Vasilis Christoforidis\\
Aristotle University of Thessaloniki\\
Archimedes, Athena Research Center\\
\texttt{vgchrist@csd.auth.gr}
}
\date{}

\newtheorem{theorem}{Theorem}
\newtheorem{lemma}{Lemma}

\newtheorem{corollary}{Corollary}

\newtheorem{proposition}{Proposition}

\begin{document}

\maketitle

\input{abstract}

\input{intro}

\input{contributions}

\input{related}

\input{prelims}

\input{binaryXOS}

\input{subadditive_chores}

\input{submodular_chores}

\input{acks}

\bibliographystyle{alpha}
\bibliography{sample}

\end{document}

%% file: abstract.tex
\begin{abstract}
    We study the approximability of envy-free up to any item (EFX) allocations for indivisible chores under complement-free cost functions. Our main result is an  instance with $255$ agents and $764$ chores with binary XOS cost functions, in which no $\alpha$-EFX allocation exists for any $1\le\alpha<2$.
    The construction is based on a combinatorial expansion property of binary labels, obtained using Sidon sets and Reed--Solomon codes. We formalize the main result in Lean 4.\footnote{https://github.com/vgchrist/efx-chores-lean}

    We also study the special case of three agents. We construct a six-chore instance with monotone subadditive cost functions for which no $\alpha$-EFX allocation exists for any $1\le\alpha < 2^{1/3}$, and a six-chore instance with monotone submodular cost functions for which no $\alpha$-EFX allocation exists for any $1\le\alpha<20/19$. These constructions are obtained by refining the original counterexample of \cite{CS24}.

\end{abstract}

%% file: intro.tex
\section{Introduction}

Fair division of indivisible items is a central problem at the interface of
economics and theoretical computer science. Given a set of agents with
heterogeneous preferences over bundles of items, the goal is to find an
allocation that satisfies a meaningful fairness guarantee. One of the most
studied notions in this setting is envy-freeness up to any item (EFX). For
goods, EFX requires that no agent envy another after the hypothetical removal of any single
good from the other agent's bundle. For chores, where items are undesirable, the analogous condition requires that no agent prefer another
agent's bundle after the removal of any single chore from her own bundle.

Despite substantial progress, EFX remains poorly
understood. For goods, very recent breakthrough work has
produced the first-ever counterexample for general monotone valuations via SAT-solving
\cite{akrami2026}. Mackenzie and Suzuki  subsequently obtained non-existence of approximate EFX allocations for  submodular and subadditive valuations with simpler constructions \cite{mackenzie2026}. For chores, non-existence under general monotone cost functions was shown by \cite{CS24}. The results of \cite{akrami2026} also show that any monotone total order over subsets can be realized via a submodular function. Hence, for comparison-based fairness notions such as exact EFX, counterexamples for monotone functions can be transferred to the submodular setting.
Combined with the chore counterexample of \cite{CS24}, this already yields exact EFX non-existence for submodular, and hence subadditive, chore costs.
The remaining issue addressed in this note is w.r.t. approximation: how large an approximation gap can be obtained by explicit constructions?

We address this question by providing novel constant-factor inapproximability results for chores.
Our main result gives a lower bound for binary XOS cost functions. We
construct an instance with $255$ agents and $764$ chores in which no
$\alpha$-EFX allocation exists for any $1\le\alpha<2$. This strengthens the
recent exact non-existence result for binary XOS costs of \cite{LinLTZ26}. We use the
same two-clause cost template $c_i(S)=\max\{z_i(S),o_i(S)\}$ as
\cite{LinLTZ26}, where $z_i(S)$ and $o_i(S)$ count the zero and one bits in
coordinate $i$ of the chore labels in $S$. We replace their finite matching
gadget by an algebraic pair expansion construction. 

For the special case of three agents, we construct a six-chore instance with
monotone subadditive cost functions for which no $\alpha$-EFX allocation
exists for any $1\le\alpha<2^{1/3}$. This narrows the gap with the known upper
bound of $2$ for three agents with subadditive cost functions
\cite{afshinmehr2024}. We also give a weighted-coverage, and hence submodular,
construction with no $\alpha$-EFX allocation for any $1\le\alpha<20/19$.
Both results complement the existential implications of \cite{akrami2026} and \cite{CS24} by exhibiting small instances with nontrivial multiplicative gaps.

The three-agent constructions are based on the counterexample of
\cite{CS24}, but the role of complementarities is removed through a compression argument. We formulate a general transfer lemma for chores:
if a monotone cost profile has no EFX allocation and each agent's cost
function takes at most $L\ge2$ distinct values on nonempty bundles, then it
can be replaced by a normalized monotone subadditive profile with no
$\alpha$-EFX allocation for every $1\le\alpha<2^{1/(L-1)}$. This argument is
inspired by the rank-based compression approach of Mackenzie and Suzuki
\cite{mackenzie2026}. A direct application of the ideas of \cite{akrami2026} and \cite{mackenzie2026} to the original construction of \cite{CS24} yields weaker inapproximability factors. We refine the original instance to improve the factor to $2^{1/3}$.

%% file: contributions.tex
\subsection{Our Contributions}

We study the approximability of EFX allocations for indivisible chores under
complement-free cost functions. The non-existence of exact EFX allocations for general monotone cost functions is known due to
\cite{CS24}. Taken together with Proposition 1 of \cite{akrami2026}, this already implies the non-existence of exact EFX allocations for submodular, and hence, subadditive cost functions. In light of these results, our focus is on obtaining explicit lower bounds.

Our main contribution is a constant-factor inapproximability result for
XOS chores. We show how to construct an instance with $255$ agents and $764$ chores with  binary XOS cost functions such that no $\alpha$-EFX allocation exists
for any $1\le \alpha<2$. We use standard notions from finite fields, coding theory, and additive combinatorics.
The proof is based on a pair expansion property of binary chore labels, which we obtain using coding-theoretic constructions. In particular, the number of agents corresponds to the number of coordinates in the resulting binary code, and our construction therefore inherently uses many agents.

For the special case of three agents, we provide a six-chore instance with monotone subadditive cost functions such that no $\alpha$-EFX allocation exists for any $1 \le \alpha < 2^{1/3}$. 
This narrows the gap with the known upper bound
of $2$ for three agents with subadditive cost functions due to \cite{afshinmehr2024}. The construction is obtained by refining the ordinal structure of the counterexample of \cite{CS24}, and applying a rank-compression argument for chores, inspired by the recent approach of \cite{mackenzie2026}. We also give an explicit lower bound for three agents with submodular costs. Specifically, we realize the same ordinal profile by monotone weighted coverage cost functions. This yields a three-agent, six-chore instance with no $\alpha$-EFX allocation for any $1 \le \alpha < 20/19$.

%% file: related.tex
\subsection{Related Work}

The existence of EFX allocations under additive valuations has been a central question in discrete fair division,  both for goods and chores. For goods, the question remains open, whereas for chores it was very recently resolved in the negative by \cite{HeTao26}. EFX was originally introduced by \cite{teac/CaragiannisKMPS19}. Since then, a large body of work has established positive results along two main directions: instances with only a few agents and instances with restricted valuation classes.
We refer the interested reader to the survey \cite{aij/AmanatidisABFLMVW23:survey} for a broader overview.

\paragraph{EFX for goods.}
Exact EFX allocations are known to exist in several special cases. These include instances with a small number of agents, such as two agents with general monotone valuations \cite{/siamdm/PlautR20} and three agents with additive valuations \cite{/jacm/ChaudhuryGM24}, with subsequent extensions to more general classes
\cite{aaai/BergerCFF22,ior/AkramiACGMM25}. Positive results are also known for several restricted valuation domains, including binary valuations and variants thereof, see, e.g., \cite{tcs/AmanatidisBFHV21:stories, aaai/BabaioffEF21:dichotomous,BSY23}, lexicographic preferences \cite{aaai/HosseiniSVX21:lex}, leveled valuations \cite{ipl/CC25},  instances with two or three valuation types \cite{dam/Mahara23:types,ec/HVGN025:types}, and graph instances \cite{ec/CFKS25:graphs}. On the approximation side, $0.618$-EFX allocations are known for additive valuations \cite{tcs/AmanatidisMN20}, with improved guarantees in restricted settings \cite{aamas/MarkakisS23,ec/AmanatidisFS24:frontier}. For subadditive valuations, $1/2$-EFX allocations always exist \cite{/siamdm/PlautR20}. Recently, \cite{akrami2026} disproved the existence of EFX allocations under general monotone valuations via SAT-solving, while 
\cite{mackenzie2026} obtained compact explicit counterexamples for submodular and subadditive valuations.

\paragraph{EFX for chores.} The existence of EFX allocations under additive cost functions was a major open problem until He and Tao recently showed that such allocations need not exist \cite{HeTao26}. Their counterexample applies for every number of agents $n\ge4$, even with tri-valued additive costs. The existence of exact EFX allocations for three agents with additive cost functions remains open. Exact existence is known in several restricted settings, including instances with (almost) identical orderings \cite{www/Li22:chores,tcs/KobayashiMS25}, additive leveled cost functions \cite{teac/GafniHLT23:copies}, and instances with two types of chores \cite{aamas/AzizLRS:chores}.

On the approximation side, \cite{aij/ZhouW24} obtained the first nontrivial guarantees for additive costs, while \cite{CS24,afshinmehr2024} subsequently derived improved results for the special case of three agents and other restricted settings. More recently, \cite{stoc/GargMQ25} obtained in a breakthrough result a constant-factor approximation  for additive cost functions, subsequently improved to $2$-EFX in \cite{aaai/GargM26}. The work of \cite{CS24} established the non-existence of EFX allocations under general monotone cost functions. 

Binary marginal costs are particularly relevant to our work. For binary cancelable costs, \cite{tcs/TaoWYZ25} showed that EFX allocations exist; the same work established the existence of $2$-EFX allocations for binary submodular costs. More recently, \cite{LinLTZ26} showed that exact EFX allocations need not exist for binary XOS and binary supermodular costs. Their binary XOS profile uses the same two-clause form $c_i(S)=\max\{z_i(S),o_i(S)\}$ that we use below. Our larger algebraic label family rules out every approximation factor below $2$.

%% file: prelims.tex
\section{Preliminaries}
Let $N=[n]$ be a set of agents and let $M$ be a finite set of indivisible
items, with $m=|M|$. We write $[n]=\{1,\dots,n\}$. An allocation is a partition $X=(X_1,\dots,X_n)$ of $M$, where $X_i$ is
the bundle assigned to agent $i$.

For goods, each agent $i\in N$ has a valuation function
$v_i:2^M\to\mathbb R_{\ge 0}$. We assume throughout that valuations are
normalized, i.e., $v_i(\emptyset)=0$, and monotone, i.e., $v_i(S)\le v_i(T)$ whenever
$S\subseteq T$. For $\alpha\in(0,1]$, an allocation $X$ is $\alpha$-EFX for
goods if, for every pair of agents $i,j\in N$ and every good $g\in X_j$, it
holds that $v_i(X_i)\ge \alpha\, v_i(X_j\setminus\{g\})$. The case
$\alpha=1$ is exact EFX.

For chores, each agent $i\in N$ has a cost function
$c_i:2^M\to\mathbb R_{\ge 0}$. We again assume normalization,
$c_i(\emptyset)=0$, and monotonicity, meaning that $c_i(S)\le c_i(T)$ whenever
$S\subseteq T$. For $\alpha\ge 1$, an allocation $X$ is $\alpha$-EFX for chores
if, for every pair of agents $i,j\in N$ and every chore $e\in X_i$, it holds
that $c_i(X_i\setminus\{e\})\le \alpha\, c_i(X_j)$. The case $\alpha=1$ is
exact EFX. We sometimes write $S\setminus e$ as shorthand for $S \setminus \{e\}$.

We focus on the following standard classes of set functions. A monotone set
function $f:2^M\to\mathbb R_{\ge 0}$ is subadditive if
$f(S)+f(T) \ge f(S\cup T)$ for all $S,T\subseteq M$. It is XOS if there exist $k$ nonnegative additive functions $f^1,\dots,f^k$ such that $f(S)=\max_{j\in[k]}f^j(S)$ for every set $S \subseteq M$, i.e., the maximum over finitely many additive functions. The class of XOS functions is equivalent to fractionally subadditive functions \cite{siamcomp/Feige09}.
It is submodular if $f(S)+f(T) \ge f(S \cup T) + f(S\cap T)$ for all $S,T\subseteq M$. Equivalently, it is submodular if for
all $S\subseteq T\subseteq M$ and all $e\in M\setminus T$, it satisfies the
decreasing marginal property
$
f(S\cup\{e\})-f(S)\ge f(T\cup\{e\})-f(T)$. A monotone set function has binary marginals if every marginal value $f(S\cup\{e\})-f(S)$ belongs to $\{0,1\}$. In particular, every binary XOS function has binary marginals.
We use the term valuation function for goods and cost function for chores.

\paragraph{Binary vector spaces, finite fields, and codes.}
Let $\mathbb F_2=\{0,1\}$ denote the field with two elements, with addition modulo $2$. For a prime power $q$, let $\mathbb F_q$ denote the finite field with $q$
elements. In particular, $\mathbb F_{2^r}$ is an $r$-dimensional vector
space over $\mathbb F_2$. After fixing an $\mathbb F_2$-basis, every element
of $\mathbb F_{2^r}$ can therefore be represented uniquely by an $r$-bit
vector. For a field $F$, we write $F[T]$ for the polynomial ring over $F$. We use the standard fact that a nonzero polynomial of degree at most $d$ over a field has at most $d$ roots.

For a binary vector $u\in\mathbb F_2^m$, define its support by $\mathrm{supp}(u)=\{i\in[m]:u(i)=1\}$ and its Hamming weight by $\mathrm{wt}(u)=|\mathrm{supp}(u)|$. A binary linear code of length $m$ is a linear subspace $C\le\mathbb F_2^m$. If $W\le C$ is a binary subspace, define $\mathrm{supp}(W)=\bigcup_{w\in W}\mathrm{supp}(w)$. An $s$-dimensional binary vector space contains exactly $2^s$ vectors, of which $2^s-1$ are nonzero. 

We also use a standard pair uniqueness notion from additive combinatorics. A
subset $A$ has the \emph{Sidon property} if sums of distinct unordered pairs are unique: whenever
$a\neq b$, and
$a+b=c+d,$
it follows that
$\{a,b\}=\{c,d\}.$

%% file: binaryXOS.tex
\section{Main Impossibility Result}
\label{sec:binary-xos}

In this section, we present our main impossibility result. 

\begin{restatable}{theorem}{binaryxosthm}
\label{thm:binary-xos}
There exists a construction with $255$ agents and $764$ chores in which every cost function is the maximum of two additive functions with binary marginals, such that no $\alpha$-EFX allocation exists for any
\(1\le\alpha<2.\) 
\end{restatable}

The construction represents each chore by a binary word of length $255$, with one coordinate for each agent. For a chore $g$, let $v_g\in\{0,1\}^{255}$ denote its binary label. For an agent $i$ and a bundle $S$, let
$z_i(S)=|\{g\in S:v_g(i)=0\}|$ and
$o_i(S)=|\{g\in S:v_g(i)=1\}|$, and define
$c_i(S)=\max\{z_i(S),o_i(S)\}$.
Thus, each $c_i$ is the maximum of two additive functions with coefficients in $\{0,1\}$.

For two chores $g,h$, define their disagreement set by
$D(g,h)=\{i\in[255]:v_g(i)\neq v_h(i)\}$.
For a pair $\{g,h\}$, agent $i$ assigns cost $1$ precisely when $i\in D(g,h)$, and cost $2$ otherwise.

\paragraph{Proof overview.}
The main combinatorial ingredient is a family of $764$ binary labels with the following expansion property: every collection of $p\le 204$ pairwise disjoint chore pairs disagrees in more than $p$ coordinates. To see why this implies the theorem, suppose that an $\alpha$-EFX allocation exists for some $\alpha<2$. A simple counting argument shows that every bundle has size between $2$ and $7$. Let $P$ be the set of agents receiving exactly two chores, and write $p=|P|$. The same counting argument gives $1\le p\le 204$.

Now consider an agent $i\notin P$. Since $|X_i|\ge 3$, two chores in $X_i$ have the same $i$-th bit. Removing a suitable third chore leaves both of them in the bundle, and hence the residual cost of agent $i$ is at least $2$. Therefore, every size-two bundle must have cost $2$ to agent $i$: if some pair bundle had cost $1$, the EFX inequality would require $2\le\alpha<2$, a contradiction. It follows that the two labels in every pair bundle agree in every coordinate outside $P$.
There are exactly $p$ pair bundles, and they are pairwise disjoint. Hence all disagreement coordinates of these $p$ pairs lie inside the $p$ coordinates in $P$. The pair expansion property rules this out, since any $p\le 204$ pairwise disjoint chore pairs must disagree in more than $p$ coordinates.

It remains to construct a family of labels satisfying this property. The construction has two components. First, a map $x\mapsto(x,x^3)$ that ensures that distinct unordered chore pairs give distinct difference vectors. Second, a finite-field evaluation code guarantees that every low-dimensional subspace generated by such differences has large coordinate support. Combining these two properties yields the required property.

\subsection{The main combinatorial tool}

We first state the required expansion property and show how it rules out any $\alpha$-EFX allocations for $1 \le \alpha <2$. We then construct labels with this property.

\begin{lemma}[Pair expansion]
\label{lem:pair-expansion}
    There exist 764 binary labels $v_g \in \{0,1\}^{255}$ such that, for every collection $\mathcal P$ of $p$ pairwise disjoint unordered pairs of chores, where $1 \le p \le 204$,
    \[ \left| \bigcup_{{g,h}\in\mathcal P} D(g,h) \right|>p.
    \] 
\end{lemma}

\begin{proposition}
\label{prop:expansion-implies-lb}
Suppose that $764$ binary labels $v_g\in\{0,1\}^{255}$ satisfy the pair expansion property of Lemma~\ref{lem:pair-expansion}. Then there exists an instance with $255$ agents and $764$ chores, in which every cost function is an XOS function with binary marginals, that admits no $\alpha$-EFX allocation for any $\alpha<2$.
\end{proposition}

\begin{proof}
    There is one agent for each coordinate $i \in [255]$. For a bundle $S \subseteq M$, define 
    \[ z_i(S)=\bigl|\{g\in S:v_g(i)=0\}\bigr|, \qquad o_i(S)=\bigl|\{g\in S:v_g(i)=1\}\bigr|, \]
    and let 
    \( c_i(S)=\max\{z_i(S),o_i(S)\}. \) Both $z_i$ and $o_i$ are binary additive functions, so $c_i$ is a two-clause binary XOS cost function. For a pair of chores $\{g,h\}$, we have
    \[ c_i(\{g,h\})= \begin{cases} 1, & \text{if } v_g(i)\neq v_h(i),\\ 2, & \text{if } v_g(i)=v_h(i). \end{cases} \]
    Equivalently, $i \in D(g,h)$ if and only if $c_i(\{g,h\})=1$.

    Fix $\alpha<2$ and suppose for contradiction that $X=(X_1,\dots, X_{255})$ is an $\alpha$-EFX allocation. We first show that every bundle has size between $2$ and $7$. 

    \textbf{No bundle has size $0$ or $1$.} Suppose that there exists an agent $j\in N$ for which $X_j = \emptyset$. The other $254$ bundles contain all $764$ chores, therefore there exists an agent $i$ with a bundle of size at least $\lceil \frac{764}{254}\rceil=4$. For every $e \in X_i$, the bundle $X_i \setminus \{e\}$ is nonempty and therefore has positive cost, while $c_i(X_j)=0$, contradicting $\alpha$-EFX. 

    Similarly, if $|X_j|=1$, then again some agent $i$ has $|X_i| \ge 4$. After hypothetically removing any chore $e \in X_i$, at least three chores remain. Among three bits, at least two are the same, hence $c_i(X_i \setminus \{e\}) \ge 2$, while $c_i(X_j)=1$, leading to a contradiction.

    Therefore, every bundle has size at least two.

    \textbf{There is a pair bundle, and no bundle has size eight or more.} We know that every bundle has size at least two. Since $764<3 \cdot 255$, at least one bundle has size exactly two; call it a pair bundle. We show that no bundle can have size eight or more: if $X_j$ is a pair bundle, suppose that there exists some $X_i$ with $|X_i| \ge 8$, then after hypothetically removing any $e \in X_i$, at least seven chores remain, so one bit occurs at least four times. Therefore, we have $c_i(X_i \setminus \{e\}) \ge 4$, for any $e \in X_i$, while $c_i(X_j) \le 2$, so $\alpha c_i(X_j) \le2\alpha<4$, contradicting $\alpha$-EFX. Hence, $2 \le \lvert X_i \rvert \le 7$, for every player $i \in N$.

    \textbf{At most 204 agents own pair bundles.} Let $P=\{i:\lvert X_i \rvert =2\}$ and $L = \{i: \lvert X_i \rvert \ge 3\}$. Let $p = \lvert P \rvert$. Recall that each agent has at least two chores, hence $2\cdot255=510$ in total. There are $254$ remaining chores. Every agent outside $P$ contributes at most $5$ more chores, therefore $5|L| \ge 254$, so $|L| \ge 51$. It follows that $p \le 255-|L| \le 204$.

    \textbf{The final contradiction.} We have established thus far that $2 \le \lvert X_i \rvert \le 7$, for all $i\in N$ and also, $1 \le p \le 204$.

    Fix some agent $i \in L$. Since $\lvert X_i \rvert \ge 3$, two chores in $X_i$ have the same $i$-th bit. Choose $e \in X_i$ different from these two chores. The two equal-bit chores remain, so by construction of the cost function, we have $c_i(X_i \setminus \{e\}) \ge 2$, for some $e \in X_i$. 

    Now let $X_j=\{g,h\}$ be a pair bundle of some agent $j$. If $v_g(i) \ne v_h(i)$, then by construction, $c_i(X_j)=1$. Hence, we would have that $2 \le c_i(X_i\setminus \{e\})\le \alpha c_i(X_j)=\alpha<2$, which is a contradiction. Thus $v_g(i)=v_h(i)$ for every $i \in L$, or equivalently, $D(g,h) \subseteq P$, for every pair bundle $\{g,h\}$.

    There are exactly $p$ pair bundles, one for each agent in $P$, and they are pairwise disjoint since $X$ is an allocation. But all their disagreements lie inside the $p$ coordinates in $P$, so
    \begin{equation}\label{eq:pair-owner-support}\left| \bigcup_{\{g,h\}\text{ pair bundle}} D(g,h) \right| \le |P| =p.\end{equation} 
    On the other hand, we have $1\le p\le 204$, so the pair expansion property of Lemma~\ref{lem:pair-expansion} applies to these $p$ pairwise disjoint pairs and gives the strict reverse inequality, contradicting~\eqref{eq:pair-owner-support}. Therefore no $\alpha$-EFX allocation exists for any $\alpha<2$.
        
\end{proof}

Note that a $2$-EFX allocation exists in this instance. Partition all $764$ chores into $254$ triplets and one pair. Every bundle is nonempty, so for every pair of agents $i,j$, we have $c_i(X_j) \ge 1$. Moreover, after removing any chore $e \in X_i$, it follows that $c_i(X_i\setminus \{e\}) \le 2$. Hence, the allocation is $2$-EFX.

\subsection{Constructing the binary labels}

We now construct the 764 binary labels. The construction uses two finite fields. Let $K=GF(2^{10})$ and $F=GF(2^4)$. The field $K$ has $2^{10}=1024$ elements and is a $10$-dimensional  vector space over \(\mathbb F_2\); similarly, \(F\) is a \(4\)-dimensional vector space over \(\mathbb F_2\). After fixing an \(\mathbb F_2\)-basis for each field, elements of \(K\) and \(F\) can be represented by \(10\)-bit and \(4\)-bit vectors, respectively. With respect to these representations, field addition is coordinatewise addition over \(\mathbb F_2\), i.e., bitwise XOR.

Choose \(764\) distinct elements \(x\in K\), one for each chore, and denote by \(x_g\) the element associated with chore \(g\).

\paragraph{A pair uniqueness gadget.}
For each chosen \(x\in K\), define $q_x=(x,x^3)\in K^2,$
where \(x^3\) is computed using the field multiplication in \(K\). Since each element of \(K\) is represented by \(10\) bits, \(q_x\) consists of two \(10\)-bit blocks. The set
$\{(x,x^3):x\in K\}\subseteq K^2
$ has the following Sidon property: the sum of two distinct points uniquely determines the corresponding unordered pair.

\begin{lemma}[Pair uniqueness]
    If $x \ne y$ and $(x,x^3)+(y,y^3)=(z,z^3)+(w,w^3)$, then $\{x,y\}=\{z,w\}$.
\end{lemma}

\begin{proof}
Let $a=x+y=z+w$. Since $x\neq y$ and the field has characteristic two, $a\neq0$. The identity $x^3+y^3=(x+y)^3+xy(x+y)$ gives $a^3+axy=a^3+azw$, and hence $xy=zw$. Therefore both $\{x,y\}$ and $\{z,w\}$ are the roots of $T^2+aT+xy$. A quadratic polynomial has at most two roots, so the unordered pairs coincide.
\end{proof}

\paragraph{Polynomial representation and extended evaluation.}
Fix an \(\mathbb F_2\)-linear isomorphism
$\psi:K^2\longrightarrow F^5.$ For each chosen $x\in K$, write $\psi(x,x^3)=(a_0(x),a_1(x),a_2(x),a_3(x),a_4(x))$ and associate with $x$ the polynomial
\[
f_x(T)=a_0(x)+a_1(x)T+a_2(x)T^2+a_3(x)T^3+a_4(x)T^4\in F[T].
\]

Thus, the map \(\psi\) identifies the \(20\)-dimensional binary vector space \(K^2\) with the space of polynomials over \(F\) of degree at most \(4\). For such a polynomial
$f(T)=a_0+a_1T+a_2T^2+a_3T^3+a_4T^4,$
we evaluate \(f\) at every element \(t\in F\) and introduce one additional formal evaluation position \(\infty\), defined by
$f(\infty)=a_4.$
Hence \(f\) has \(17\) extended evaluation coordinates, indexed by
$F\cup\{\infty\}.$

\begin{lemma}
Let \(f\in F[T]\) be a nonzero polynomial of degree at most \(4\). Then
$
\bigl|\{t\in F\cup\{\infty\}:f(t)\neq0\}\bigr|
\ge 13.
$

\end{lemma}

\begin{proof}
If \(\deg f=4\), then \(a_4\neq0\), so \(f(\infty)\neq0\). Moreover, a nonzero polynomial of degree \(4\) over a field has at most \(4\) roots in \(F\). Hence at most \(4\) of the \(17\) extended evaluation coordinates are zero.

If \(\deg f\le3\), then \(a_4=0\), so \(f(\infty)=0\), while \(f\) has at most \(3\) roots in \(F\). Again, at most \(4\) of the \(17\) extended evaluation coordinates are zero. Thus \(f\) is nonzero at at least $17-4=13$ positions.
\end{proof}

This is the standard extended Reed--Solomon evaluation of polynomials of degree at most $4$ over $F$.

\paragraph{Binary encoding.} We next convert each evaluation value $f(t)\in F$ into binary coordinates. Since $F$ is a $4$-dimensional vector space over $\mathbb F_2$, there are  $2^4=16$ $\mathbb F_2$-linear maps \( \ell:F\longrightarrow\mathbb F_2, \) of which $15$ are nonzero. For every $t\in F\cup\{\infty\}$ and every nonzero such map $\ell$, define \( E(f)_{t,\ell}=\ell(f(t)). \) Thus $E(f)\in\mathbb F_2^{17\cdot 15}=\mathbb F_2^{255}$. Evaluation and each $\ell$ are $\mathbb F_2$-linear, so $E$ is $\mathbb F_2$-linear.

\begin{lemma}
    For every nonzero $a\in F$, exactly eight of the fifteen nonzero $\mathbb F_2$-linear maps $\ell:F\to\mathbb F_2$ satisfy $\ell(a)=1.$
\end{lemma}

\begin{proof}
    Since $a\neq 0$, there exists an $\mathbb F_2$-linear map $\ell_0:F\to\mathbb F_2$ such that $\ell_0(a)=1$. Pair every linear map $\ell$ with $\ell+\ell_0$. The two maps in each pair take opposite values on $a$. Hence exactly half of the $16$ linear maps, namely $8$, satisfy $\ell(a)=1$. Since the zero map takes value $0$ on $a$, all eight of these maps are nonzero.
\end{proof}

\begin{lemma}[Minimum weight and injectivity]
    Every nonzero polynomial $f\in F[T]$ of degree at most $4$ satisfies $\mathrm{wt}(E(f))\ge 104.$ In particular, $E$ is injective.
\end{lemma}

\begin{proof}
    By the extended evaluation bound, $f(t)\neq 0$ at at least $13$ of the $17$ evaluation positions. At each such position, exactly eight of the associated $15$ binary coordinates are equal to $1$. Therefore $\mathrm{wt}(E(f))\ge13\cdot8=104$.

    In particular, no nonzero polynomial is mapped to the all-zero vector. Since $E$ is linear, if $E(f)=E(g)$, then $ E(f-g)=0,$ and hence $f-g=0$. Therefore $f=g$, so $E$ is injective.
\end{proof}

Finally, for each chore represented by $x\in K$, define its binary label by \( v_x=E(f_x)\in\mathbb F_2^{255}. \)

\paragraph{From minimum weight to pair expansion.}
The map $x\mapsto (x,x^3)$ guarantees uniqueness of pair sums, while the encoding $E$ ensures that every nonzero codeword has large Hamming weight.

Let \( C:=\{E(a_0+a_1T+\cdots+a_4T^4):a_0,\ldots,a_4\in F\}\subseteq \mathbb F_2^{255}. \) By linearity of $E$, the set $C$ is a binary linear code. It follows that uniqueness in $K^2$ also transfers to $\mathbb F_2^{255}$.

\begin{corollary}[Distinct pair differences]
    For distinct two-element subsets $\{x,y\}$ and $\{z,w\}$ of the chosen field elements,  $v_x+v_y\neq v_z+v_w.$ Moreover, if $x\neq y$, then $ v_x+v_y\neq 0. $

\end{corollary}

\begin{proof}
The map from $K^2$ to binary labels is linear and injective. Equality of pair differences would therefore contradict pair uniqueness. Injectivity also ensures that distinct field elements have distinct labels.
\end{proof}

\begin{lemma}[Support bound]
    If $W\le C$ is an $s$-dimensional binary subspace, then $|\mathrm{supp}(W)|\ge 208(1-2^{-s}).$
\end{lemma}

\begin{proof}
The case $s=0$ is immediate, so suppose $s\ge1$. The space $W$ contains $2^s-1$ nonzero vectors, each of weight at least $104$. For every $i\in\mathrm{supp}(W)$, the coordinate map $w\mapsto w(i)$ is a nonzero linear functional on $W$, so exactly $2^{s-1}$ vectors take value $1$ in that coordinate. Double-counting the ones gives
\[
104(2^s-1)\le\sum_{w\in W}\mathrm{wt}(w)
=2^{s-1}|\mathrm{supp}(W)|.
\]
Dividing by $2^{s-1}$ yields the claimed bound.
\end{proof}

We are now ready to prove Lemma~\ref{lem:pair-expansion}.

\begin{proof}[Proof of Lemma~\ref{lem:pair-expansion}]
Let $\mathcal P$ be a collection of $p$ pairwise disjoint unordered pairs of chores, where $1\le p\le 204$, and let
\(
W=\mathrm{span}\{v_x+v_y:\{x,y\}\in\mathcal P\}.
\) Write $s=\dim W$. 
By the pair uniqueness property, the $p$ vectors $v_x+v_y$ are distinct and nonzero. Since an $s$-dimensional binary vector space contains exactly $2^s-1$ nonzero vectors, we have $p\le 2^s-1$.

Moreover, the support of $W$ is exactly the union of the supports of these generators. Indeed, each generator belongs to $W$, while a binary linear combination cannot be nonzero in a coordinate in which every generator is zero. Hence
\[
\mathrm{supp}(W)
=
\bigcup_{\{x,y\}\in\mathcal P}\mathrm{supp}(v_x+v_y)
=
\bigcup_{\{x,y\}\in\mathcal P}D(x,y).
\]

It remains to show that $|\mathrm{supp}(W)|>p$. If $s\le 7$, then by the support bound,
\[
|\mathrm{supp}(W)|
\ge 208(1-2^{-s})
=
\frac{208}{2^s}(2^s-1)
>
2^s-1
\ge p,
\]
where the strict inequality follows from $2^s<208$.

If $s\ge 8$, then
\[
|\mathrm{supp}(W)|
\ge 208(1-2^{-s})
\ge 208\left(1-\frac1{256}\right)
>204
\ge p.
\]

Thus in all cases,
\[
\left|\bigcup_{\{x,y\}\in\mathcal P}D(x,y)\right|
=
|\mathrm{supp}(W)|
>p,
\]
which proves the lemma.
\end{proof}

More broadly, any $\alpha$-EFX allocation with $\alpha<2$ would force the $p$ pair bundles to disagree only on the $p$ coordinates corresponding to their owners. The pair expansion property rules this out; every collection of $p\le 204$ pairwise disjoint chore pairs disagrees on more than $p$ coordinates. Hence no such approximate EFX allocation exists.

\binaryxosthm*

\begin{proof}
By Lemma~\ref{lem:pair-expansion}, the labels constructed above satisfy the pair expansion property. The claim therefore follows immediately from Proposition~\ref{prop:expansion-implies-lb}.
\end{proof}

%% file: subadditive_chores.tex
\section{Three Agents}
\subsection{Subadditive Cost Functions}
\label{sec:subadditiveChores}

In this section, we obtain EFX inapproximability for three agents with  monotone subadditive cost functions. Our starting point is the counterexample of \cite{CS24}. Although the original construction uses complementarities, the main obstruction is fundamentally ordinal. We exploit this by replacing the original cardinal costs with compressed rank values. Our main tool is the chore analogue of the ordinal compression lemma of \cite{mackenzie2026}. The resulting functions are automatically subadditive, while the strict comparisons needed to violate EFX are preserved with a multiplicative gap. We improve the bound obtained from the direct application of the ordinal
compression lemma to the counterexample of \cite{CS24}. Rather than preserving all cost levels of the original construction, 
we refine it to reduce the number of ordinal levels. This yields a stronger inapproximability factor.

The following ordinal compression lemma is the main transfer tool of the section.

\begin{lemma}
\label{lemma:compression}
    Let $M$ be a finite set of chores, and let $d_1,\dots,d_n$ be arbitrary monotone cost functions over $2^M$. Suppose that the profile $d=(d_1,\dots,d_n)$ admits no EFX allocation. Assume further that, for each agent $i$, the collection of nonempty bundles takes at most $L\ge2$ distinct $d_i$-values (i.e., the set $\{d_i(S): \emptyset \neq S \subseteq M\}$ has cardinality at most $L$). Then, there exist normalized monotone subadditive cost functions $c_1, \dots, c_n$ over $M$ such that, for every $1 \le \alpha < 2^{1/(L-1)}$, the profile $c=(c_1,\dots,c_n)$ admits no $\alpha$-EFX allocation.
\end{lemma}

\begin{proof}
Let $\lambda=2^{1/(L-1)}$. Fix an agent $i$. Let $L_i$ be the number of
distinct values taken by $d_i$ on nonempty bundles, and enumerate these values
as $t_{i,0}<t_{i,1}<\cdots<t_{i,L_i-1}$. Thus $1\le L_i\le L$. For every
nonempty bundle $S$, define its rank $\rho_i(S)$ by
$d_i(S)=t_{i,\rho_i(S)}$, i.e., $\rho_i(S)=r \iff d_i(S)=t_{i,r}$. Define $c_i(\emptyset)=0$, and for every nonempty
$S\subseteq M$, set $c_i(S)=\frac12\lambda^{\rho_i(S)}$.

The function $c_i$ is normalized by definition. It is monotone because if
$S\subseteq T$ and $S\neq\emptyset$, then $T\neq\emptyset$ and
$d_i(S)\le d_i(T)$, hence $\rho_i(S)\le \rho_i(T)$ and therefore
$c_i(S)\le c_i(T)$. If $S=\emptyset$, monotonicity is immediate. 
We next prove subadditivity. For every nonempty $S$, we have
$\frac12\le c_i(S)\le \frac12\lambda^{L_i-1}\le \frac12\lambda^{L-1}=1$.
Hence, if $A$ and $B$ are both nonempty, then
$c_i(A\cup B)\le 1\le c_i(A)+c_i(B)$. If one of $A,B$ is empty, the
subadditivity inequality is immediate. Therefore $c_i$ is subadditive.

It remains to prove the transfer of the EFX violation and inapproximability. Let $X$ be an arbitrary
allocation. Since $d$ admits no EFX allocation, there exist agents $i,j$ and
a chore $e\in X_i$ such that $d_i(X_i\setminus e)>d_i(X_j)$.  
If $X_j=\emptyset$, then the strict inequality implies $(X_i\setminus e)\neq\emptyset$, and
therefore $c_i(X_i\setminus e)>0=\alpha c_i(X_j)$. Thus $X$ violates $\alpha$-EFX. 
Now suppose $X_j\neq\emptyset$. The strict inequality also implies that
$X_i\setminus e$ is nonempty. Since $d_i(X_i\setminus e)>d_i(X_j)$, we have
$\rho_i(X_i\setminus e)\ge \rho_i(X_j)+1$. Therefore
$c_i(X_i\setminus e)=\frac12\lambda^{\rho_i(X_i\setminus e)}\ge
\frac12\lambda^{\rho_i(X_j)+1}=\lambda c_i(X_j)>\alpha c_i(X_j)$,
where the last inequality uses $\alpha<\lambda$. Thus $X$ violates
$\alpha$-EFX. Since $X$ was arbitrary, no allocation is $\alpha$-EFX for any
$1\le \alpha<\lambda=2^{1/(L-1)}$.
\end{proof}

\paragraph{Warm-Up: A Direct Application to the CS24 Counterexample.}

Applying Lemma~\ref{lemma:compression} directly to the \(k=3\)
construction of \cite{CS24}, which has seven nonempty cost levels, gives a
subadditive instance with no \(\alpha\)-EFX allocation for any
\(\alpha<2^{1/6}\). We provide a short exposition below.

We have three agents and six chores $M=\{h, \ell_1, \ell_2, b_1, b_2, b_3\}$. Let $A=\{h, \ell_1, \ell_2\}$, $B=\{b_1, b_2, b_3\}$, and $B_{-i}=B\setminus \{b_i\}$. For a parameter $k > 2$, define the original cost functions $d_i$ as follows: $d_i(h)=k$, $d_i(\ell_1)=d_i(\ell_2)=1$, and $d_i(b_1)=d_i(b_2)=d_i(b_3)=0$, for all agents $i\in N$. For any bundle $S$ with $\lvert S \rvert \ge 2$, it holds \[d_i(S)=\begin{cases}
    k^2, & B_{-i} \subseteq S \text{ or } (b_i\in S \text{ and } S \cap A \neq \emptyset),\\
    \sum_{x \in S}d_i(x), & \text{otherwise}
\end{cases}\] 

This instance does not admit an EFX allocation; in fact, under the original superadditive costs, the approximation ratio grows unbounded as $k \to \infty$.
We set $k=3$; then, the possible values of nonempty bundles of each $d_i$ are contained in $\{0,1,2,3,4,5,9\}$. Indeed, penalized bundles have cost $9$, while non-penalized bundles have additive cost determined only by the chores in $A$. Thus, each agent has at most $L=7$ distinct cost levels. Explicitly, the complete cost function is given by $c_i(\emptyset)=0,$ 
and, for every nonempty bundle \(S\subseteq M\),
\[
c_i(S)=
\begin{cases}
2^{-1}, & d_i(S)=0,\\
2^{-5/6}, & d_i(S)=1,\\
2^{-4/6}, & d_i(S)=2,\\
2^{-3/6}, & d_i(S)=3,\\
2^{-2/6}, & d_i(S)=4,\\
2^{-1/6}, & d_i(S)=5,\\
1, & d_i(S)=9.
\end{cases}
\]

\paragraph{An Improved Inapproximability Bound.}
The direct application above preserves more levels than  necessary. We
now isolate the ordinal structure  and compress it
to four levels.

Let \(M=\{h,\ell_1,\ell_2,b_1,b_2,b_3\}\). Let
\(A=\{h,\ell_1,\ell_2\}\), \(L=\{\ell_1,\ell_2\}\), and
\(B=\{b_1,b_2,b_3\}\). For each agent \(i\in\{1,2,3\}\), define a
cost function \(d_i:2^M\to\{0,1,2,3\}\) by
\[
d_i(S)=
\begin{cases}
3, & B\setminus\{b_i\}\subseteq S,\\
3, & b_i\in S\text{ and }S\cap A\neq\emptyset,\\
2, & h\in S\text{ and neither level-\(3\) condition holds},\\
1, & S\cap L\neq\varnothing\text{ and neither level-\(3\) condition holds},\\
0, & \text{otherwise.}
\end{cases}
\]
The level-\(3\) bundles are the triggered bundles. The remaining levels
distinguish bundles containing the heavy ordinary chore \(h\), bundles
containing only light ordinary chores, and safe bundles containing only special
chores.

\begin{lemma}
\label{lem:four-level-chores-no-efx}
The profile \(d=(d_1,d_2,d_3)\) admits no EFX allocation.
\end{lemma}

\begin{proof}
Let \(X=(X_1,X_2,X_3)\) be an arbitrary allocation. We show that \(X\) violates
EFX.

First suppose that some agent receives at least two chores from \(B\). By
symmetry, assume this agent is agent \(1\).

Suppose first that \(B\setminus\{b_1\}=\{b_2,b_3\}\subseteq X_1\). If
\(X_1\neq\{b_2,b_3\}\), choose \(e\in X_1\setminus\{b_2,b_3\}\). Then
\(d_1(X_1\setminus\{e\})=3\). Among \(X_2\) and \(X_3\), at least one bundle
has \(d_1\)-cost at most \(2\): only one of them can contain \(b_1\), and
neither can contain both \(b_2\) and \(b_3\). Thus agent \(1\) strongly envies
one of these bundles after removing \(e\).
If instead \(X_1=\{b_2,b_3\}\), then \(d_2(X_1)=d_3(X_1)=0\). The remaining
chores \(A\cup\{b_1\}\) are allocated to agents \(2\) and \(3\). One of these
agents receives at least two chores; since only one remaining chore is \(b_1\),
that bundle contains a chore from \(A\), and after deleting a suitable chore, the remaining bundle still has positive cost. Hence that agent strongly envies agent
\(1\), whose bundle has cost \(0\) for her.

It remains to consider the case in which agent \(1\) receives \(b_1\) and one
other special chore. By symmetry, assume \(\{b_1,b_3\}\subseteq X_1\). If
\(X_1\cap A\neq\varnothing\), then deleting \(b_3\) leaves \(b_1\) together with
a chore from \(A\), so \(d_1(X_1\setminus\{b_3\})=3\). The other two bundles
contain neither \(b_1\) nor both \(b_2,b_3\), so both have \(d_1\)-cost at most
\(2\). Thus EFX fails.
Now suppose \(X_1=\{b_1,b_3\}\). Then \(d_3(X_1)=0\), and the remaining chores
are \(A\cup\{b_2\}\). If agent \(3\) receives at least two of these remaining
chores, then after deleting a suitable chore her remaining bundle has positive
cost, so she strongly envies agent \(1\). Hence, in any EFX allocation, agent
\(3\) would have to receive at most one remaining chore. If agent \(3\) receives
\(b_2\), then agent \(2\) receives all three chores in \(A\), and deleting one
light chore leaves a positive-cost residual while \(d_2(X_3)=d_2(b_2)=0\).
If agent \(3\) does not receive \(b_2\), then agent \(2\) receives \(b_2\) and
at least two chores from \(A\); deleting one chore from \(A\) leaves \(b_2\)
together with another chore from \(A\), a level-\(3\) bundle for agent \(2\),
while \(X_3\) has \(d_2\)-cost at most \(2\). In both cases, EFX fails.

Therefore, no EFX allocation can have an agent receiving at least two chores
from \(B\). Hence, in any remaining candidate allocation, every agent receives
exactly one chore from \(B\).

Now suppose that some agent receives her matching special chore. By symmetry,
assume \(b_1\in X_1\). If \(|X_1\cap A|\ge 2\), deleting one chore from \(A\)
leaves \(b_1\) together with another chore from \(A\), a level-\(3\) bundle
for agent \(1\). The other two bundles split \(b_2\) and \(b_3\), contain no
\(b_1\), and therefore both have \(d_1\)-cost at most \(2\). Thus EFX fails.
If \(X_1=\{b_1\}\), then the three chores in \(A\) are split between agents
\(2\) and \(3\). One of these agents receives at least two chores from \(A\),
and hence has a positive-cost residual after deleting a suitable chore. For
that agent, \(X_1\) has cost \(0\), so EFX fails.
If \(X_1=\{b_1,h\}\), then deleting \(b_1\) leaves \(\{h\}\), which has level
\(2\) for agent \(1\). The other two bundles contain no \(h\), no \(b_1\), and
do not contain both \(b_2,b_3\), so both have \(d_1\)-cost at most \(1\). Hence
EFX fails.
Finally, suppose \(X_1=\{b_1,\ell\}\) for some \(\ell\in L\). Let
\(r\in\{2,3\}\) be the agent who receives \(h\). Agent \(r\)'s bundle contains
exactly one chore from \(B\). Deleting that \(B\)-chore leaves a residual bundle
containing \(h\), and hence of cost at least \(2\) for agent \(r\). On the other
hand, \(X_1\) has cost \(1\) for agent \(r\), since it contains one light
ordinary chore and is not a triggered bundle. Thus EFX fails.

Therefore, in any EFX allocation, no agent receives her matching special chore.
Since every agent receives exactly one chore from \(B\), it follows that each
agent receives a special chore \(b_j\) with \(j\neq i\).
Let \(r\) be the agent who receives \(h\). Agent \(r\)'s bundle contains a
nonmatching special chore. Delete this special chore. The remaining bundle
contains \(h\), and hence has \(d_r\)-cost at least \(2\). Among the two other
bundles, one contains neither \(h\) nor \(b_r\). Since every bundle contains
exactly one chore from \(B\), that bundle also cannot contain both chores in
\(B\setminus\{b_r\}\). Therefore it has \(d_r\)-cost at most \(1\). Agent \(r\)
strongly envies that bundle after deleting her nonmatching special chore.

All allocations violate EFX. Hence, the profile \(d\) admits no EFX
allocation.
\end{proof}

We now instantiate the lemma using the refined version above.

\begin{theorem}
\label{thm:subadditive-2-13}
An \(\alpha\)-EFX allocation need not exist for three agents with monotone
subadditive cost functions for any \(1\le \alpha<2^{1/3}\).
\end{theorem}

\begin{proof}
By Lemma~\ref{lem:four-level-chores-no-efx}, the four-level profile satisfies Lemma~\ref{lemma:compression} with $L=4$. Therefore, the compressed costs are normalized, monotone, subadditive, and no $\alpha$-EFX exists for $1\le \alpha <2^{1/3}$. The empty bundle is mapped to $0$, while nonempty bundles of levels
\(0,1,2,3\) are mapped to \(2^{-1},2^{-2/3},2^{-1/3},1\), respectively.
\end{proof}

We note that the construction easily extends to every $m\ge6$ by adding dummy chores and 
letting $c_i'(S)=c_i(S \cap M)$
, where $M$ is the original chore set. Then \(c_i'\) remains normalized, monotone, and subadditive. Hence the same
inapproximability bound holds for all \(m\ge 6\).

%% file: submodular_chores.tex
\subsection{Submodular Cost Functions}

We now switch to submodular cost functions. The non-existence of EFX allocations in this case also follows from \cite{CS24} and \cite{akrami2026}, since it is shown that monotone comparison orders over bundles can be realized by monotone submodular set functions. Since EFX is comparison-based, this immediately implies that the non-existence example of \cite{CS24} for monotone costs can be converted into exact non-existence for submodular costs. This observation settles the non-existence of EFX allocations for the classes considered here, but is purely existential. We give a weighted-coverage realization of the construction obtained in Section~\ref{sec:subadditiveChores}, yielding an improved lower bound for submodular cost functions.

We present the following proposition here for completeness, stated and proved as
Proposition~1 in \cite{akrami2026} and attributed there to Bhaskar Ray
Chaudhury and Uriel Feige.

\begin{proposition}
\label{prop:Akrami et al.}[Chaudhury, Feige; Proposition~1 of \cite{akrami2026}]
    Let $M$ be a finite set and let $\prec$ be a total ordering on $2^M$ satisfying $S \subset T \implies S \prec T$. Then there exists a submodular function $f:2^M \to \mathbb R_{\ge 0}$ such that $S\prec T \implies f(S) < f(T)$.
\end{proposition}

That is, any monotone total ordering of
bundles can be realized by a monotone submodular set function while preserving
all strict comparisons. Although Proposition~\ref{prop:Akrami et al.} is stated for goods, it is a
statement about set functions and therefore applies equally to cost functions
for chores.

\begin{corollary}
\label{cor:submodular-chores-no-efx}
There exists a three-agent, six-chore instance with monotone submodular cost
functions that admits no EFX allocation.
\end{corollary}

We now present a submodular realization of the ordinal profile that relies on the construction obtained in Section~\ref{sec:subadditiveChores}. Recall that applying Proposition~\ref{prop:Akrami et al.} of \cite{akrami2026} does establish the non-existence of EFX allocations under submodular cost functions, but it does not yield by itself a large multiplicative gap. We present a tighter construction below, which implements the four-level ordinal skeleton of the modified \cite{CS24} counterexample by weighted coverage functions.

We next show that the refined four-level obstruction can be realized by
weighted coverage functions. 

Let \(M=\{h,\ell_1,\ell_2,b_1,b_2,b_3\}\), let
\(A=\{h,\ell_1,\ell_2\}\), \(L=\{\ell_1,\ell_2\}\), and
\(B=\{b_1,b_2,b_3\}\). For each agent \(i\), write
\(B_{-i}=B\setminus\{b_i\}\).
Recall the four-level ordinal profile:
\[
d_i(S)=
\begin{cases}
3, & B_{-i}\subseteq S,\\
3, & b_i\in S\text{ and }S\cap A\neq\emptyset,\\
2, & h\in S\text{ and neither level-\(3\) condition holds},\\
1, & S\cap L\neq\emptyset\text{ and neither level-\(3\) condition holds},\\
0, & \text{otherwise.}
\end{cases}
\]
The level-\(3\) bundles are the triggered bundles: either the bundle contains
both special chores in \(B_{-i}\), or it contains \(b_i\) together with an
ordinary chore.
For \(R\subseteq M\), let \(\chi_R(S)=1\) if \(S\cap R\neq\varnothing\), and
\(\chi_R(S)=0\) otherwise. For each agent \(i\), define
\[
c_i(S)=
7\sum_{b\in B_{-i}}\chi_{A\cup\{b\}}(S)
+\sum_{b\in B_{-i}}\chi_{\{b_i,b\}}(S)
+2\sum_{b\in B_{-i}}\chi_{\{h,b_i,b\}}(S).
\]
This is a weighted coverage function. The atoms \(A\cup\{b\}\) detect ordinary
chores, the atoms \(\{h,b_i,b\}\) separate bundles containing \(h\) from bundles
containing only light ordinary chores, and the atoms \(\{b_i,b\}\) separate the
triggered special bundles from the lower levels.

\begin{lemma}
\label{lem:coverage-realization}
For every agent \(i\), the function \(c_i\) is monotone submodular. Moreover,
for all bundles \(S,T\subseteq M\), if \(d_i(S)>d_i(T)\), then
\(c_i(S)\ge \frac{20}{19}c_i(T)\).
\end{lemma}

\begin{proof}
Each atom \(\chi_R\) is a coverage function: it records whether the bundle
intersects \(R\). Hence each \(\chi_R\) is monotone submodular. Since \(c_i\) is
a nonnegative linear combination of such atoms, it is monotone submodular.
It remains to prove the separation claim. We compute the possible values of
\(c_i(S)\) by ordinal level.

If \(d_i(S)=3\), then \(S\) intersects all six coverage atoms. Indeed, if
\(B_{-i}\subseteq S\), then each atom indexed by a chore in \(B_{-i}\) fires.
If \(b_i\in S\) and \(S\cap A\neq\varnothing\), then the ordinary chore from $A$ hits
the two atoms \(A\cup\{b\}\), while \(b_i\) hits the remaining four atoms.
Therefore \(c_i(S)=2(7+1+2)=20\).

If \(d_i(S)=2\), then \(h\in S\), \(b_i\notin S\), and \(S\) contains at most
one chore from \(B_{-i}\). The two atoms \(A\cup\{b\}\) and the two atoms
\(\{h,b_i,b\}\), for \(b\in B_{-i}\), always fire due to $h$. If \(S\) also contains one
chore from \(B_{-i}\), then one additional atom of the form \(\{b_i,b\}\) fires.
Hence \(c_i(S)\in\{18,19\}\).

If \(d_i(S)=1\), then \(S\) contains at least one light chore, contains no
\(h\), contains no \(b_i\), and contains at most one chore from \(B_{-i}\). The
two atoms \(A\cup\{b\}\), for \(b\in B_{-i}\), always fire. If \(S\) also
contains one chore from \(B_{-i}\), then the corresponding atoms
\(\{b_i,b\}\) and \(\{h,b_i,b\}\) also fire. Hence \(c_i(S)\in\{14,17\}\).

Finally, if \(d_i(S)=0\), then \(S\) is either empty or a safe special-only
bundle: \(\{b_i\}\), \(\{b\}\) for \(b\in B_{-i}\), or \(\{b_i,b\}\) for
\(b\in B_{-i}\). Direct evaluation gives \(c_i(S)\in\{0,6,10,13\}\).

Thus the possible values by level are:
\[
\begin{array}{c|c}
d_i(S) & \text{possible values of }c_i(S)\\
\hline
0 & 0,\ 6,\ 10,\ 13\\
1 & 14,\ 17\\
2 & 18,\ 19\\
3 & 20
\end{array}
\]
The adjacent level gaps are at least
$14/13, 18/17, 20/19$.
The minimum of these ratios is $20/19$. Therefore, whenever
\(d_i(S)>d_i(T)\), we have \(c_i(S)\ge (20/19)c_i(T)\).
\end{proof}

\begin{theorem}
\label{thm:coverage-20-19}
There exists a three-agent, six-chore instance with monotone weighted-coverage
cost functions such that no \(\alpha\)-EFX allocation exists for any
\(1\le \alpha<20/19\).
\end{theorem}

\begin{proof}
By Lemma~\ref{lem:four-level-chores-no-efx}, the profile
\(d=(d_1,d_2,d_3)\) admits no EFX allocation. Hence, for every allocation
\(X=(X_1,X_2,X_3)\), there exist agents \(i,j\) and a chore \(e\in X_i\) such
that \(d_i(X_i\setminus e)>d_i(X_j)\). If $X_j$ is the empty set, the claim is immediate. Otherwise, if $X_j \neq \emptyset$, by
Lemma~\ref{lem:coverage-realization}, this implies
\(c_i(X_i\setminus e)\ge (20/19)c_i(X_j)\). Therefore, for every
\(\alpha<20/19\), the allocation \(X\) violates \(\alpha\)-EFX. Since \(X\) was
arbitrary, no \(\alpha\)-EFX allocation exists.
\end{proof}

%% file: acks.tex
\paragraph{AI Disclosure.}
The author used ChatGPT to assist with exploratory proof search, including proposing proof strategies and potential counterexamples. The construction of Section \ref{sec:binary-xos} was developed by OpenAI’s GPT-5.6 Sol. We additionally used Aristotle to produce a Lean~4 formalization of the main construction and its proof, as well as of the three-agent results. The resulting formalization was checked with Mathlib. All results were independently checked, substantially rewritten, and verified by the human author. The author takes full responsibility for all claims, proofs, references, and final text.